\documentclass [14pt]{amsart}
\usepackage{amsmath}
\usepackage{amssymb}
\usepackage{amscd}
\usepackage{epsfig}
\usepackage{amsxtra}
\usepackage{pifont}
\usepackage{mathabx}
\usepackage{MnSymbol}
\usepackage{accents}
\usepackage{scalerel,stackengine}
\usepackage{xcolor}
\allowdisplaybreaks
\stackMath
\newcommand\reallywidehat[1]{%
\savestack{\tmpbox}{\stretchto{%
  \scaleto{%
    \scalerel*[\widthof{\ensuremath{#1}}]{\kern-.6pt\bigwedge\kern-.6pt}%
    {\rule[-\textheight/2]{1ex}{\textheight}}
  }{\textheight}%
}{0.5ex}}%
\stackon[1pt]{#1}{\tmpbox}%
}

\newcommand\reallywidecheck[1]{%
\savestack{\tmpbox}{\stretchto{%
  \scaleto{%
    \scalerel*[\widthof{\ensuremath{#1}}]{\kern-.6pt\bigwedge\kern-.6pt}%
    {\rule[-\textheight/2]{1ex}{\textheight}}
  }{\textheight}%
}{0.5ex}}%
\stackon[1pt]{#1}{\scalebox{-1}{\tmpbox}}%
}

\numberwithin{equation}{section}

\newcommand{\RR}{{\mathbb R}}

\newcommand{\CC}{{\mathbb C}}

\newcommand{\TT}{\mathbb T}

\newcommand{\cA}{{\mathcal A}}
\newcommand{\cB}{{\mathcal B}}

\newcommand{\dd}{\mbox{\rm d}}

\newcommand{\Cc}{C_{\mathsf{c}}}

 \newtheorem{theorem}{Theorem}[section]
 \newtheorem{lemma}[theorem]{Lemma}
 \newtheorem{prop}[theorem]{Proposition}
 \newtheorem{coro}[theorem]{Corollary}

 \newtheorem{definition}[theorem]{Definition}
 \newtheorem{example}[theorem]{Example}
  \newtheorem{remark}[theorem]{Remark}

\newcommand{\lb}{\text{\textlquill} }
\newcommand{\rb}{\text{\textrquill} }

\allowdisplaybreaks

\begin{document}
\title[Eberlein convolution as inner product]{Eberlein convolution as an inner product and quantitative Bombieri--Taylor results}
\author{Daniel Lenz}
\address{Mathematisches Institut, Friedrich Schiller Universit\"at Jena, 07743 Jena, Germany}
\email{daniel.lenz@uni-jena.de}
\urladdr{http://www.analysis-lenz.uni-jena.de}

\author{Nicolae Strungaru}
\address{Department of Mathematics and Statistics, MacEwan University \\
10700 -- 104 Avenue, Edmonton, AB, T5J 4S2, Canada\\
and \\
Institute of Mathematics ``Simon Stoilow''\\
Bucharest, Romania}
\email{strungarun@macewan.ca}
\urladdr{https://sites.google.com/macewan.ca/nicolae-strungaru/home}

\begin{abstract} We present a simple argument providing a very general Bombieri--Taylor-type result to compute Bragg peaks in diffraction theory  as limits. At its abstract core the argument relies on considering the Eberlein convolution as an inner product. This inner product  gives an abstract version of a  Cauchy--Schwarz type inequality which implies the Bombieri--Taylor-type result as well as  generalizations.
\end{abstract}

\maketitle

\section{Introduction}
This article deals with a basic question in mathematical diffraction theory. Mathematical diffraction theory has become a flourishing field of research in the last decades not least due to  the discovery
of quasicrystals some four decades ago by diffraction experiments. This discovery was later honored with a Nobel prize in Chemistry in 2011. A key feature of quasicrystals is pure point diffraction with symmetries excluding lattices. The underlying order is known as aperiodic order in mathematics, see e.g. the  monographs and survey collections  \cite{BG,BG2,KLS}  for details and further references.

A basic framework of mathematical diffraction theory, as developed by Hof in \cite{Hof},  starts with a Delone  set $\varLambda$  in  Euclidean space $\RR^d$ modeling  the  positions of the atoms of the solid in question. Here, Delone set means that  there exist $r,R>0$ such that any ball of size $r$ in Euclidean space meets at most one point of $\varLambda$ and any ball of radius $R$ meets at least one point of $\varLambda$.

The autocorrelation $\gamma$  of the solid along a sequence $(B_n)$ of balls around the origin in $\RR^d$ with radii going to infinity comes about by an averaging procedure as
$$\gamma = \gamma_{(B_n)} = \lim_{n\to \infty} \frac{1}{|B_n|} \sum_{x,y\in \varLambda \cap B_n} \delta_{x-y} \,,$$
where the limit is assumed to exist  in the sense of vague convergence of measures and  $\delta_z$ denotes the unit point mass at $z\in\RR^d$.

The Fourier transform
$\widehat{\gamma}$
of $\gamma$ is called diffraction measure and  governs the outcome of diffraction experiments for the solid. In line with the discovery of quasicrystals, the case where $\widehat{\gamma}$ is a pure point measure (or has a  large point component) is of prime interest. In particular one is interested in  the atoms of $\widehat{\gamma}$, which are  called Bragg peaks of $\varLambda$,   and the value assigned to Bragg peaks by $\widehat{\gamma}$, which are  called the intensity of the Bragg peaks.

The basic intuition about the Bragg peaks and their intensities is that
\begin{equation}\label{eq:BT}
A_y = \lim_{n\to \infty} \frac{1}{|B_n|}\sum_{x\in \varLambda\cap B_n} e^{- 2 \pi i  x \cdot y } \mbox{
exists  and } \widehat{\gamma}(\{y\}) = |A_y|^2 \mbox{
holds }\tag{BT}
\end{equation}
for all $y \in\RR^d$.  In particular the standard method to determine the Bragg peaks is by computing $A_y$ and
singling out those $y$ with $A_y \neq 0$. This  requires care and effort
as  neither existence of the limit nor the equality are  true in general. Accordingly, they have  been  subject to intense research for various specific models.

 Both conditions in \eqref{eq:BT}  were assumed to hold  in  the very first articles on the topic by Bombieri and Taylor \cite{BT,BT2}, without any justification and the issue then  became dubbed 'Bombieri--Taylor  conjecture' by Hof in \cite{Hof}.  For substitution models this  was then justified in \cite{GK} and for cut-and-project models in \cite{Schlottmann}. A connection to dynamical system and continuity of eigenfunction was explored in various places and tackled in a rather general situation in \cite{Lenz}  to which we refer for further details and references.   A weaker version of \eqref{eq:BT} in a rather general context was then established in \cite{LS}.  In the physical literature, \eqref{eq:BT} is also known under the name 'commutativity of the Wiener diagram', see e.g.  \cite{BG} for discussion. In the influential survey of Lagarias \cite{Lag} this  is called the 'consistent phase property'. Further  discussion can be found in the authors' work with Spindeler \cite{LSS}.

Below we present a very simple argument showing that the   inequality
$$\widehat{\gamma}(\{y\}) \geq |\widetilde{A}_y|^2$$
 holds for all $y \in\RR^d$ without \textbf{any} further assumption, where
$$\widetilde{A_y} =\limsup_{n\to \infty}\frac{1}{|B_n|} \left| \sum_{x\in \varLambda\cap B_n} e^{- 2 \pi  i x \cdot y }\right| \,.$$
This justifies all previous literature on the topic arguing that  non-vanishing $A_y$ give Bragg peaks of the systems.

The simple argument is based on an approach developed  in a series of articles \cite{LSS,LSS2,LSS3}. Indeed, this approach allows us to infer the  desired inequality as a consequence of  (a suitable abstraction of) a Cauchy--Schwarz type inequality.  Relevant  parts are summarized in the next two  sections. The actual argument is then given in Section \ref{sec-BT}. An further abstraction recovering a main result of \cite{LS} is the contained in Section \ref{sec-CS}.

As it does not present any complication but rather makes the arguments more transparent, we work below in the setting of a locally compact abelian group $G$ (rather than just Euclidean space $\RR^d$) and consider translation bounded measures $\mu$ (rather than Delone sets $\varLambda$).

\section{Background and notation}\label{sec-back}
In this section we   recall some basic  notation and definitions from harmonic analysis (see \cite{ARMA1,BF} for details).

Let $G$ be a locally compact abelian group (LCAG).   The group operations are written additively and
the neutral element is denoted by $0$.   The Haar measure on $G$ is denoted by $\theta_G$ and
integration of $f$ with
respect to Haar measure is denoted by $\int_G f \dd s$. Also, we write $|A|$ instead of $\theta_G (A)$ for  measurable subsets of $G$.
We denote by $\Cc(G)$ the  vector space consisting of compactly supported continuous functions on $G$. For a $\varphi \in \Cc (G)$ we define $\widetilde{\varphi} \in \Cc (G)$ by
$\widetilde{\varphi}(t) := \overline{\varphi(-t)}$ for $t\in G$.

A \emph{Radon measure} on  $G$  is a linear functional $\mu : \Cc(G)
\to \CC$ with the property that for each compact set $K \subseteq G$
there exists  $C_K\geq 0$ such that all functions $\varphi \in
\Cc(G)$ whose support is contained in $K$ satisfy
\[
\left| \mu(\varphi) \right| \leq C_K \| \varphi \|_\infty \,.
\]
We will often write $\int_{G} \varphi(t) \dd \mu(t):= \mu(\varphi)$.

Given a Radon measure $\mu$, we can define a new Radon measure $\widetilde{\mu}$ via
\[
\widetilde{\mu}(\varphi) := \overline{\mu(\tilde{\varphi})} \qquad \forall \varphi \in \Cc(G) \,.
\]

To any Radon measure $\mu$ there exists a unique positive regular measure
$\nu$ and a measurable $h : G\longrightarrow \CC$ with $|h| =1$
such that
\[
\mu (\varphi) = \int_G \varphi(t) \cdot h(t) \dd \nu(t)
\]
holds  for all $\varphi \in\Cc (G)$ \cite{Ped}. The measure $\nu$ is called
the \textit{total variation} of $\mu$ and henceforth denoted by $|\mu|$.  For suitable measurable  functions $f$ on $G$, e.g. those that are bounded and vanish outside a compact set we then define
$$\int_G f(t) \dd\mu(t) :=\int_G f(t) \cdot h(t) \dd|\mu|(t) \,.$$
Note that this then implies
$$\int_G f(t) \cdot  \overline{h(t)} \dd\mu(t) = \int_{G} f(t) d|\mu|(t)$$
for all such $f$ due to $|h| =1$.

The measure $\mu$ is called \emph{finite} if $|\mu|(G)<\infty$ holds.

Whenever $A$ is a Borel subset of $G$ we define the restriction $\mu|_A$ of $\mu$ to
$A$ to be the measure satisfying
\[
\mu|_A (\varphi) :=\int_A \varphi(t) \cdot h(t) \dd|\mu|(t)
\]
for all $\varphi \in\Cc (G)$.

\smallskip

The convolution $\varphi \ast\psi$ of $\varphi, \psi\in\Cc (G)$ is
the function on $G$ defined by
\[
\varphi \ast \psi (t) := \int_{ G} \varphi (s) \cdot \psi (t-s) \dd s \,.
\]
The convolution  $\mu\ast \varphi$ between  a
measure $\mu$ and a function $\varphi \in \Cc(G)$ is the function on
$G$ defined by
\[
\mu\ast \varphi (t) = \int_{G} \varphi(t-s) \dd \mu(s)\,.
\]

If the measure $\mu$ has the property that $\mu\ast \varphi$ is  bounded for all $\varphi \in\Cc (G)$ it is called \textit{translation bounded}.

The convolution $\mu\ast \nu$ between two finite measures is the
measure given by
\[
\mu \ast\nu (\varphi) = \int_{G} \int_{ G} \varphi (s+t) \dd\mu(s) \dd\nu
(t) \qquad \mbox{ for }  \varphi \in \Cc(G) \,.
\]

A measure $\gamma$ is called \textit{positive definite} if
\[
\gamma\ast \varphi \ast \widetilde{\varphi} (0)\geq 0
\]
holds for all $\varphi \in\Cc (G)$.

The dual group $\widehat{G}$ of $G$ is the set of all continuous group homomorphisms $\chi :  G\longrightarrow \TT$. Here, $\TT$ is the group of  complex numbers with modulus $1$ (equipped with multiplication).  The dual group $\widehat{G}$  is a locally compact Abelian group in a natural way. The \textit{Fourier transform} $\widehat{f}$ of $f \in L^1(G)$ is the function on $\widehat{G}$ with
\[
\widehat{f}(\chi) = \int_{G} \overline{\chi(t)} \cdot  f (t) \dd t \,.
\]

\medskip

Every positive definite measure $\gamma$ admits a (unique) positive  measure $\widehat{\gamma}$ on $\widehat{G}$ with
\[
\gamma\ast \varphi \ast \widetilde{\varphi}(0) = \int_{\widehat{G}} |\widehat{\varphi}(\chi) |^2 \dd\widehat{\gamma}(\chi)
\]
for all  $\varphi \in \Cc (G)$. The measure $\widehat{\gamma}$ is called the \textit{Fourier transform} of $\gamma$ (see \cite{BF} for details).

\begin{example}[$\chi \theta_G$]\label{ex-chi} For us the measure $ \gamma =  \chi\theta_G$ for $\chi \in \widehat{G}$ will play a special role.  It is easily seen to be positive definite and to satisfy $\gamma = \reallywidehat{\delta_\chi}$, where $\delta_\chi$ denotes the unit point measure at $\chi \in\widehat{G}$.
\end{example}

\begin{definition} The linear hull of the positive definite measures on $G$ is denoted by $LP (G)$.
\end{definition}
We note that the Fourier transform can be extended from the positive definite measures to all of $LP(G)$ by linearity. We write $\widehat{\mu}$ for the Fourier transform of $\mu\in LP(G)$.

\section{The Eberlein convolution as inner product}
In this section we present  the framework of diffraction via the Eberlein convolution as  developed in
\cite{LSS,LSS2,LSS3}. We refer to these works for details and proofs.

We say that $\mu, \nu$ have a well defined \textit{reflected Eberlein convolution} with respect to the van Hove net $\cA$ if
\[
\lb \mu , \nu \rb_{\cA} := \lim_{i} \frac{1}{|A_i|} (\mu|_{A_i})*\widetilde{(\nu|_{A_i})}
\]
exists. In this case, we call $ \lb  \mu, \nu \rb_{\cA}$ the \textit{reflected Eberlein convolution} of $\mu$ and $\nu$.

If $\lb \mu, \mu \rb_{\cA}$ exists, then it is called the \textit{autocorrelation measure of $\mu$} and is denoted by $\gamma_{\mu}$.

The Eberlein convolution has many features of an inner product. In particular the following holds for  translation bounded measures  $\mu, \nu,\eta$:

\begin{itemize}
\item[(IP1)]If $\lb \nu,\nu\rb_{\cA} $ exists it is positive definite.

\item[(IP2)]  If $ \lb \mu , \nu
\rb_{\cA}$ exists, then,  $ \lb \nu , \mu
\rb_{\cA}$ exists and $
 \lb  \nu,  \mu \rb_{\cA} = \widetilde{ \lb \mu, \nu \rb_{\cA} }$
holds.

\item[(IP3)]  If $\lb \mu,\nu\rb_{\cA} $  and  $\lb \mu,\eta\rb_{\cA}$ exist then $\lb\mu, a \nu + b\eta\rb_{\cA}$ exists as well and
\[
\lb\mu, a \nu + b\eta\rb_{\cA} = a \lb \mu,\nu\rb_{\cA} + b\lb \mu,\eta\rb_{\cA}
\]
holds for all $a,b\in\CC$.

\end{itemize}

As consequence of $(IP2),(IP3)$ is a polarization-type identity. Specifically,
\begin{equation}\label{eq:pol}
 \lb\mu,\nu\rb_{\cA} = \sum_{k=0}^4 i^k \cdot  \lb \mu + i^k \nu,\mu + i^k \nu\rb_{\cA} \tag{Pol}
\end{equation}
holds for translation bounded $\mu,\nu$ whenever the autocorrelations of $\mu$ and $\nu$ exist and $\lb\mu,\nu\rb_{\cA}$ exists.  So, we infer in particular that $\lb \mu,\nu\rb_{\cA}$ belongs to the $LP(G)$ in this case. In particular, $ \lb\mu,\nu\rb_{\cA}$ is Fourier transformable as a measure.  We will come back to this in a later section of the article.

To shorten notation, we will often denote by
\begin{align*}
  \gamma_{\mu,\nu} &:=  \lb\mu,\nu\rb_{\cA} \\
  \rho_{\mu,\nu} &:= \reallywidehat{\gamma_{\mu,\nu}}=\reallywidehat{ \lb\mu,\nu\rb_{\cA}}
\end{align*}

\smallskip

Let us next recall the concept of Fourier--Bohr coefficient.  Specifically, for a translation bounded measure $\mu$ and a $\chi \in\widehat{G}$ we call
\[
a_{\chi}^\cA(\mu) = \lim_{i} \frac{1}{|A_i|} \int_{A_i} \overline{\chi(t)} \dd \mu(t) \,.
\]
the \textit{Fourier-Bohr coefficient} of $\mu$ at $\chi$ (if the limit exists).
The Fourier-Bohr coefficient arises from an Eberlein convolution. Indeed,  the Fourier--Bohr coefficient
$a_{\chi}^\cA(\mu)$ exists  if and only if  $\lb \mu, \chi\theta_{G} \rb_{\cA}$ exists and
$$\lb \mu, \chi \theta_{G} \rb_{\cA}= (a_{\chi}^\cA(\mu) \chi)\theta_{G}
$$
holds.  Furthermore, for  all $\chi \in \widehat{G}$ we have
\[
\lb \chi \theta_{G}, \chi \theta_{G} \rb_{\cA}= \chi \theta_{G} \,.
\]

\section{The quantitative Bombieri--Taylor}\label{sec-BT}
In this section we use the set-up given in the previous two sections to provide a quantitative Bombieri--Taylor result.

\begin{lemma}\label{lem-BT} Let $\cA=\{ A_i\}$ be  a van Hove net and  $\mu$  a translation bounded  measure and $\chi \in \widehat{G}$. Assume that the autocorrelation $\gamma_{\mu}$ and the Fourier--Bohr coefficient $a^\cA_{\chi}(\mu)$ exists with respect to $\cA$.
Let $c \in \CC$ be arbitrary and set
\[
\nu:= \mu -(c\chi) \theta_{G} \,.
\]
Then, the autocorrelation $\gamma_{\nu}$ exists with respect to $\cA$ and
\[
\gamma_{\nu}= \gamma_{\mu}+ \left( \left| c- a_{\chi}^\cA(\mu) \right|^2- \left|a_{\chi}^\cA(\mu) \right|^2  \right) \chi\theta_{G} \,.
\]
In particular,
\[
\reallywidehat{\gamma_{\nu}}(\{ \chi \}) = \reallywidehat{\gamma_{\mu}}(\{ \chi \})+\left| c- a_{\chi}^\cA(\mu) \right|^2- \left|a_{\chi}^\cA(\mu) \right|^2  \,.
\]
\end{lemma}
\begin{proof} By the inner product properties (IP2), (IP3) we see that
\[
\gamma_\nu = \lb \mu -(c\chi) \theta_{G} ,\mu -(c\chi) \theta_{G} \rb_\cA
\]
exists and satisfies
\begin{align*}
 \gamma_{\nu} &=\lb \mu -(c\chi) \theta_{G} ,\mu -(c\chi) \theta_{G} \rb_\cA \\
   &=\lb \mu ,\mu \rb_\cA+ -\overline{c} \lb \mu ,\chi \theta_{G} \rb_\cA-c \lb \mu \chi \theta_{G} ,\mu \theta_{G} \rb_\cA+|c|^2\lb \mu \chi \theta_{G} ,\chi \theta_{G} \rb_\cA \\
   &= \gamma_{\mu}-\overline{c}  (a_{\chi}^\cA(\mu) \chi)\theta_{G} - c (\overline{a_{\chi}^\cA(\mu)} \chi)\theta_{G}+|c|^2 \chi\theta_{G} \\
   &= \gamma_{\mu}+ \left( |c|^2-\overline{c} a_{\chi}^\cA(\mu)  - c \overline{a_{\chi}^\cA(\mu)}\right) \chi\theta_{G} \,.
\end{align*}
By simple algebraic manipulations we find
\begin{align*}
|c|^2-\overline{c}  a_{\chi}^\cA(\mu)  - c \overline{a_{\chi}^\cA(\mu)} &= c \bar{c} -\overline{c}  a_{\chi}^\cA(\mu)  - c \overline{a_{\chi}^\cA(\mu)} +a_{\chi}^\cA(\mu)  \overline{a_{\chi}^\cA(\mu) } -a_{\chi}^\cA(\mu)  \overline{a_{\chi}^\cA(\mu) }  \\
 &= (c- a_{\chi}^\cA(\mu))\cdot \overline{ (c- a_{\chi}^\cA(\mu))} - a_{\chi}^\cA(\mu)  \overline{a_{\chi}^\cA(\mu) } \\
 &= \left| c- a_{\chi}^\cA(\mu) \right|^2- \left|a_{\chi}^\cA(\mu) \right|^2
\end{align*}
Putting this together we obtain the first statement of the lemma. The second statement is an immediate consequence of the first statement after taking the Fourier transform and appealing to the Example ~\ref{ex-chi} discussed above.
\end{proof}

We can now prove the main result in this paper.

\begin{theorem}[Quantitative Bombieri--Taylor] Let $\mu$ be a translation bounded measure, let $\cA=\{ A_i\}$ be a van Hove net on $G$ along which the autocorrelation $\gamma_{\mu}$ exists and  let $\chi \in \widehat{G}$. Let
\[
\widetilde{A_\chi} := \limsup_i \frac{1}{|A_i|} \left| \int_{A_i} \overline{\chi(t)} \dd \mu(t)\right| \,.
\]
Then,
\[
\reallywidehat{\gamma_{\mu}}(\chi) \geq \left(\widetilde{A_\chi}\right)^2 \,.
\]
\end{theorem}
\begin{proof}
Let $(B_j)$ be a subnet of $\cA$ such that
\[
\widetilde{A_\chi} := \lim_{j} \frac{1}{|B_j|} \left| \int_{B_j} \overline{\chi(t)} \dd \mu(t)\right| \,.
\]
Since $\mu$ is translation bounded, we have $\widetilde{A_\chi} <\infty$. Then, by compactness, there exists some subnet $(C_l)$ of $(B_j)$ along which the following limit exists
\[
A_\chi := \lim_{l} \frac{1}{|C_l|} \int_{C_l} \overline{\chi(t)} \dd \mu(t) \,.
\]
Then,
\[
\widetilde{A_\chi} = |A_\chi| \,.
\]
Since $\gamma_{\mu}$ is the autocorrelation of $\mu$ along $C_l$, we can apply Lemma~\ref{lem-BT} to obtain that for all $c \in \CC$ we have along $C_l$
\[
\reallywidehat{\gamma_{\mu-c\chi}}(\{ \chi \}) = \reallywidehat{\gamma_{\mu}}(\{ \chi \})+\left| c- A_\chi  \right|^2- \left|A_\chi  \right|^2  \,.
\]
This gives
\[
 \reallywidehat{\gamma_{\mu}}(\{ \chi \})+\left| c- A_\chi \right|^2- \left|A_\chi \right|^2= \reallywidehat{\gamma_{\mu-c\chi}}(\{ \chi \}) \geq 0
\]
and hence
\[
\left|A_\chi \right|^2 \leq  \reallywidehat{\gamma_{\mu}}(\{ \chi \})+\left| c- A_\chi \right|^2 \qquad \forall c \in \CC \,.
\]
In particular, when $c=A_\chi$ we get
\[
 \reallywidehat{\gamma_{\mu}}(\{ \chi \}) \geq \left|A_\chi \right|^2 =\left(\widetilde{A_\chi}\right)^2 \,.
\]
This proves the claim.
\end{proof}

\begin{coro}Let $\cA=\{ A_i\}$ be a van Hove net on $G$ and  $\mu$ a translation bounded measure  and $\chi \in \widehat{G}$.
Let $\cB = \{B_j\}$ be any subnet of $\cA$ along which the autocorrelation $\gamma_{\mu}$ exists and let $A_\chi$ be any cluster point of
\[
\frac{1}{|B_j|} \int_{B_j} \overline{\chi(t)} \dd \mu(t) \,.
\]
Then,
\[
\reallywidehat{\gamma_{\mu}}(\chi) \geq |A_\chi|^2 \,.
\]\qed
\end{coro}

As a consequence, we get the following way to find the Bragg peaks.

\begin{coro}\cite{LS} Let $\cA=\{ A_i\}$ a van Hove net on $G$ and $\mu$ a translation bounded measure  and $\chi \in \widehat{G}$. If the autocorrelation $\gamma_{\mu}$ exists with respect to $\cA$ and if
\[
\frac{1}{|A_j|} \int_{A_j} \overline{\chi(t)} \dd \mu(t) \not\rightarrow 0
\]
then
\[
\reallywidehat{\gamma_{\mu}}(\{\chi\})>0 \,.
\]\qed
\end{coro}

\section{A generalization via Cauchy--Schwarz-type inequality}\label{sec-CS}
In this section we have a closer look at the previous argument. This reveals a Cauchy--Schwarz-type inequality behind it. This Cauchy--Schwarz-type inequality allows us to obtain a short alternative proof for a main result of \cite{LS}.

\medskip

A linear map $\tau : LP(G)\longrightarrow \CC$ is called \textit{positive} if $\tau(\gamma)\geq 0$ holds for all positive definite $\gamma$.

\begin{lemma}[Cauchy--Schwarz]\label{lem-cs} Let $\mu,\nu$ be translation bounded measures. Assume that $\lb \mu,\mu\rb_{\cA},\lb \nu,\nu\rb_{\cA}$ and $\lb \mu,\nu\rb_{\cA}$ exist. Let $\tau : LP(G)\longrightarrow \CC$ be positive linear functional. Then,
$$\left|\tau(\lb\mu,\nu\rb_{\cA})\right|^2\leq \tau(\lb\mu,\mu\rb_{\cA})\cdot  \tau (\lb \nu,\nu\rb_{\cA}) \,.$$
\end{lemma}
\begin{proof} This is a variant of the usual proof of Cauchy--Schwarz inequality. For the convenience of the reader we include the details.  As the autocorrelations of $\mu$ and $\nu$ exist as well as $\lb \mu,\nu\rb_{\cA}$ we can use (IP3) to define
\begin{eqnarray*} p(\lambda)&:=&\tau (\lb \mu + \lambda \nu, \mu + \lambda\nu\rb_{\cA})\\
&=& \tau(\lb \mu,\mu\rb_{\cA}) + \lambda \tau(\lb \mu,\nu\rb_{\cA}) + \overline{\lambda} \tau(\lb \nu,\mu\rb_{\cA}) + |\lambda|^2 \tau(\lb\nu,\nu\rb_{\cA}).
\end{eqnarray*}
for $\lambda$ in the complex numbers. By (IP1) and positivity of $\tau$ we infer $p(\lambda)\geq 0$ for all $\lambda$. If $\tau(\lb\mu,\mu\rb_{\cA}) =0$ or $\tau (\lb\nu,\nu\rb_{\cA}) =0$ holds, we easily infer
 $\tau (\lb\mu,\nu\rb_{\cA}) =0$ by taking suitable limits of $\lambda$.

So, we can assume without loss of generality $\tau(\lb\mu,\mu\rb_{\cA})\neq 0$ and $ \tau(\lb\nu,\nu\rb_{\cA} ) \neq 0$. From this we find
$$ 0\leq p(\lambda) = \tau(\lb\nu,\nu\rb_{\cA}) \left|\lambda +\frac{\tau(\lb\nu,\mu\rb_{\cA})}{\tau(\lb \nu,\nu\rb_{\cA})}\right|^2  + \tau(\lb\mu,\mu\rb_{\cA}) -   \frac{|\tau(\lb\nu,\mu\rb_{\cA}|^2}{\tau(\lb \nu,\nu\rb_{\cA})} . $$
Choosing $\lambda$ so that the first bracket vanishes then easily gives the desired statement.
\end{proof}

For any measurable $E\subset \widehat{G}$ supported inside a compact set the functional
$$\tau_E : LP(G)\longrightarrow \CC, \tau_E (\mu) :=\widehat{\mu}(E),$$
is positive and we obtain the following consequence of the previous lemma.

\begin{coro}\label{cor1} Let $E$ be a relatively compact measurable subset of $\widehat{G}$. Let $\mu,\nu$ be translation bounded measures. Assume that $\lb \mu,\mu\rb_{\cA},\lb \nu,\nu\rb_{\cA}$ and $\lb \mu,\nu\rb_{\cA}$ exist. Then, $\rho_{\mu,\nu}:= \reallywidehat{\lb\mu,\nu\rb_{\cA}}$ satisfies
$$|\rho_{\mu,\nu} (E)|^2 \leq  \reallywidehat{\gamma_\mu} (E) \cdot \reallywidehat{\gamma_\nu} (E)\,.$$ \qed
\end{coro}

\begin{remark} This corollary captures our main result on the Quantitative Bombieri--Taylor. Indeed, with $\mu = \mu$ and $\nu = \chi \theta_G$ and assuming that $\gamma_\mu$ and $a_\chi(\mu)= \rho_{\mu,\chi}(\{\chi\})$ exist, we find
$$|a_\chi (\mu)| = |\rho_{\mu,\chi}(\{\chi\})| \leq \reallywidehat{\gamma_\mu} (\{\chi\})^{1/2}  \reallywidehat{\gamma_\chi} (\{\chi\})^{1/2}=\reallywidehat{\gamma_\mu} (\{\chi\})^{1/2}  \,.$$
If $a_\chi (\mu)$ does not exist, we can go to a subnet along which it does exist and infer the statement of the theorem.
\end{remark}

We are going to extend the statement of the previous corollary from sets  to functions. To do so, we need some preparation.

\begin{prop}\label{prop1} Let $\mu,\nu$ be translation bounded measures. Assume that $\gamma_\mu= \lb \mu,\mu\rb_{\cA}, \gamma_\nu= \lb \nu,\nu\rb_{\cA}$ and $\gamma_{\mu,\nu}=\lb \mu,\nu\rb_{\cA}$ exist.

Let $f,g$ be bounded measurable functions on $\widehat{G}$ vanishing outside a compact set. Then,
$$\left|\int_{\widehat{G}} f(\chi) \cdot g(\chi) \, \dd \rho_{\mu,\nu}(\chi) \right|^2 \leq \biggl( \int_{\widehat{G}} |f(\chi)|^2 \dd \reallywidehat{\gamma_\mu}(\chi) \biggr) \cdot \biggl(\int_{\widehat{G}}|g(\chi)|^2 \dd \reallywidehat{\gamma_\nu}(\chi)\biggr) \,.$$
\end{prop}
\begin{proof}  Let $K$ be a compact set outside of which $f$ and $g$ vanish.
Approximating $f$ and $g$ by step functions we can assume without loss of generality that there exist disjoint measurable sets $E_j \subset K$, $j=1,\ldots, N$, and $c_1,\ldots, c_N, d_1,\ldots, d_N\in \CC$ such that
$$ f =  \sum_{j=1}^N c_j 1_{E_j} \mbox{ and } g = \sum_{j=1}^N d_j 1_{E_j}.$$
Now, we can compute
\begin{align*}
\left|\int_{\widehat{G}} f(\chi) \cdot g(\chi) \, \dd \rho_{\mu,\nu}(\chi) \right|^2  =&  \left| \sum_{j=1}^N c_j d_j  \,  \rho_{\mu,\nu}(E_j)\right|^2\\
(\mbox{Corollary \ref{cor1}})\:\; &\leq  \left|\sum_{j=1}^N |c_j| \reallywidehat{\gamma_\mu}(E_j)^{1/2} |d_j|  \reallywidehat{\gamma_\nu}(E_j)^{1/2}\right|^2\\
(\mbox{Cauchy-Schwarz})\;\:&\leq \left(\sum_{j=1}^N |c_j|^2 \reallywidehat{\gamma_\mu}(E_j)\right) \left(\sum_{j=1}^N |d_j|^2 \reallywidehat{\gamma_\nu}(E_j)\right)\\
&=\biggl(\int_{\widehat{G}} |f(\chi)|^2 \dd\reallywidehat{\gamma_\mu}(\chi)\biggr) \cdot \biggl(\int_{\widehat{G}} |g(\chi)|^2 \dd\reallywidehat{\gamma_\nu}(\chi) \biggr) \,.
\end{align*}
This finishes proof.
\end{proof}

After these preparations we can now come to the main result of this section.

\begin{theorem} Let $f,g : \widehat{G} \to \CC$ be  Borel measurable.
Let $\mu,\nu$ be translation bounded measures. Assume that $\gamma_\mu= \lb \mu,\mu\rb_{\cA}, \gamma_\nu= \lb \nu,\nu\rb_{\cA}$ and $\gamma_{\mu,\nu}=\lb \mu,\nu\rb_{\cA}$ exist.
Then,
$$\left( \int_{\widehat{G}}\left|f(\chi) \cdot g(\chi)\right| \, \dd  |\rho_{\mu,\nu}|(\chi) \right)^2 \leq  \left( \int_{\widehat{G}} |f(\chi)|^2 \dd  \reallywidehat{\gamma_\mu}(\chi) \right) \cdot \left( \int_{\widehat{G}}|g(\chi)|^2  \dd \reallywidehat{\gamma_\nu}(\chi)\right) \,.$$
\end{theorem}
\begin{proof}  By  (inner) regularity of the total variation $|\rho_{\mu,\nu}|$ we can assume without loss of generality that $f$ and $g$ are bounded and vanish outside a compact set $K$.

We then find from the definition of the total variation  (see Section \ref{sec-back})
$$\int_{\widehat{G}}\left|f(\chi) \cdot g(\chi)\right| \, \dd  |\rho_{\mu,\nu}|(\chi) = \int_{\widehat{G}}\left|f(\chi) \cdot g(\chi)\right| \cdot h(\chi) \, \dd  \rho_{\mu,\nu}(\chi)$$
for some measurable $h : \widehat{G}\longrightarrow \CC$ with $|h|=1$.

Now, since $|h|=1$, the claim follows from Proposition~\ref{prop1} applied to $|f|$ and $|g| h$ (instead of  $f$ and $g$).
\end{proof}

\begin{remark} The  previous theorem deals with integration with respect to the total variation of $\reallywidehat{\lb\mu,\nu\rb_{\cA}}$. Of course, we also have an estimate without the total variation. Specifically,   for bounded $f,g$ vanishing outside a compact set we find
\begin{align*}
\left| \int_{\widehat{G}}f(\chi) \cdot g(\chi) \, \dd  \rho_{\mu,\nu}(\chi) \right|^2 &\leq \left( \int_{\widehat{G}}\left|f(\chi) \cdot g(\chi)\right| \, \dd  |\rho_{\mu,\nu}|(\chi) \right)^2 \\
&\leq  \left( \int_{\widehat{G}} |f(\chi)|^2 \dd  \reallywidehat{\gamma_\mu}(\chi) \right) \cdot \left( \int_{\widehat{G}}|g(\chi)|^2  \dd \reallywidehat{\gamma_\nu}(\chi)\right) \,.
\end{align*}
\end{remark}

An immediate consequence of the theorem is noted next.

\begin{coro}
Let $\mu,\nu$ be translation bounded measures.  Assume that $\gamma_\mu= \lb \mu,\mu\rb_{\cA}, \gamma_\nu= \lb \nu,\nu\rb_{\cA}$ and $\gamma_{\mu,\nu}=\lb \mu,\nu\rb_{\cA}$ exist. Then, $| \rho_{\mu,\nu}|$ is absolutely continuous with respect to each of $\reallywidehat{\gamma_\mu}$ and $\reallywidehat{\gamma_\nu}$. In particular, also $ \rho_{\mu,\nu}$ is absolutely continuous with respect to each of $\reallywidehat{\gamma_\mu}$ and $\reallywidehat{\gamma_\nu}$.
\qed
\end{coro}

\begin{coro}Let $\mu,\nu$ be translation bounded measures. Assume that $\gamma_\mu= \lb \mu,\mu\rb_{\cA}, \gamma_\nu= \lb \nu,\nu\rb_{\cA}$ and $\gamma_{\mu,\nu}=\lb \mu,\nu\rb_{\cA}$ exist and that the measures $\reallywidehat{\gamma_\mu}$ and $\reallywidehat{\gamma_\nu}$  are supported inside measurable subsets $A,B$ of $\widehat{G}$.
Then, $\rho_{\mu,\nu}$ is supported inside $A \cap B$. \qed
\end{coro}

In particular, we get the following result, which was already proven in \cite{LS}.

\begin{coro}Let $\mu,\nu$ be translation bounded measures.  Assume that the autocorrelations of both $\mu$ and $\nu$ exist and $\reallywidehat{\gamma_\mu}$ and $\reallywidehat{\gamma_\nu}$  are supported on disjoint measurable subsets $A,B$ of $\widehat{G}$. Then, $\lb\mu,\nu\rb_{\cA}$ exists and
\[
\lb\mu,\nu\rb_{\cA} =0 \,.
\] \qed
\end{coro}

\subsection*{Acknowledgments} N.S. was supported by the NSERC Discovery grant 2024-04853, and he is grateful for the support.

\end{document}